\documentclass[11pt,a4paper]{article}
\usepackage{jheppub}

\usepackage{amssymb,amsmath,amsthm,amscd,amsfonts,mathrsfs}
\usepackage{hyperref}

\newtheorem{thm}{Theorem}
\newtheorem{lemma}{Lemma}
\newtheorem{proposition}{Proposition}
\newtheorem{remark}{Remark}

\def\AIR{{\sc air}}
\def\FIRE{{\sc Fire}}
\def\Reduze{{\tt Reduze}}

\begin{document}
\title{Integration-by-parts identities from the viewpoint of differential 
  geometry}
\author{Yang Zhang}
\affiliation{
Niels Bohr International Academy and Discovery Center, Niels Bohr Institute, \\
University of Copenhagen, Blegdamsvej 17, DK-2100 Copenhagen, Denmark
}
\abstract{We present a new method to construct integration-by-part (IBP)
  identities from the viewpoint of
  differential geometry. Vectors for generating IBP identities are
reformulated as differential forms, via Poincar\'{e} duality. Using the
tools of differential geometry and commutative algebra, we can efficiently find differential forms which
generate on-shell IBP relation without doubled propagator. Various $D=4$ two-loop
examples are presented. }
\maketitle

\section{Introduction}
With the successful run of the Large Hadron Collider (LHC), there is
an eager demand for the next-to-leading order (NLO) and
next-to-next-to-leading order (NNLO) background computation. NLO and
NNLO computations involve loop-order Feynman diagrams. The number
of Feynman integrals grows quickly for multi-leg and multi-loop cases. However, for each
diagram, many different Feynman integrals are linearly related by the
integration-by-parts (IBP) relations or symmetries, so the whole set
of integrals
can be reduced to a minimal set of integrals, so-called {\it master
  integrals} (MIs). This paper focuses on the geometric meaning for IBP
relations and provides a new method for obtaining IBP relations.

Schematically, for a $L$-loop integral, the integration of a total derivative vanishes and
resulting identity is called an IBP relation:    
\begin{equation}
  \int \frac{d^D l_1}{i \pi^{D/2}} \ldots \frac{d^D l_L}{i \pi^{D/2}}
  \sum_{i=1}^L \frac{\partial }{\partial l_i^\mu}\bigg(\frac{v_i^\mu}{D_1^{a_1}
    \ldots D_k^{a_k}}\bigg)=0.
\label{IBP}
\end{equation}
Here $v_i^\mu$ are vectors depends on externel and internal momenta.

Traditionally, various contributions to a certain
amplitude are characterized by Feynman diagrams, and the final results
are reduced to the form of MIs by IBP relations. In recent years, there are a lot of
new methods to improve the efficiency of multi-loop diagram computation,
and most of which also require the calculation of IBP identities at
certain steps. Unitarity methods \cite{Bern:1994zx,Bern:1994cg, Britto:2004nc} relate a loop amplitude to the product
of tree amplitudes, and the latter can be efficiently calculated by
recursive methods \cite{Britto:2004ap, Britto:2005fq}. 

For example,
Ossola-Papadopoulos-Pittau (OPP) method  \cite{Ossola:2006us,Ossola:2007ax,Giele:2008ve,Badger:2008cm,Ellis:2007br,Forde:2007mi}
determines the {\it minimal integrand basis} for one-loop Feynman diagrams
algebraically via partial fraction. This method has been successfully
generalized to multi-loop integrand level reduction by computational
algebraic geometry \cite{Mastrolia:2011pr,
  Badger:2012dp,Zhang:2012ce,Mastrolia:2012an,Badger:2012dv,Badger:2013gxa,Feng:2012bm,Mastrolia:2012wf,Mastrolia:2012du,Kleiss:2012yv,Huang:2013kh,Fazio:2014xea,
  vanDeurzen:2013saa, Mastrolia:2013kca,
Hauenstein:2014mda}. The coefficients of the minimal integrand are
therefore fixed by unitarity cuts. However, usually the integrand basis
is not the minimal integral basis, so finally the results are
reduced MIs by IBP relations. Multi-loop unitarity has also been systematically performed by the
maximal unitarity method \cite{Kosower:2011ty, CaronHuot:2012ab, Larsen:2012sx,
  Johansson:2012zv, Johansson:2013sda, Sogaard:2013yga,
  Sogaard:2013fpa,Sogaard:2014ila,Sogaard:2014oka} . Feynman
integrals are converted to contour integrals and MI
coefficients can be directly extracted from residue calculations. To get
the correct contour weights, in the intermediate step, IBP relations
are required \cite{Kosower:2011ty}.

For multi-loop or multi-leg diagrams, in general, the computation of
IBP is very heavy. For a given loop diagram, there are many IBP relations from
different choices of IBP-generating vectors $v_i^\mu$ in (\ref{IBP}). The desired reduction of
Feynman integrals to
MIs can be achieved by Gaussian elimination of IBP relations, via
Laporta algorithm \cite{Laporta:2000, Laporta:2001}. This algorithm is
used for several sophisticated programs, like \AIR{}
\cite{AIR}, \FIRE{} \cite{FIRE} and \Reduze{}
\cite{Reduze}. Furthermore, Laporta algorithm can be greatly sped up
by finite fields numerical sampling method \cite{vonManteuffel:2014ixa}. 
 
A breakthrough method for generating IBP relations by Gluza, Kajda and
Kosower (GKK method) \cite{Gluza:2010ws}, appeared in 2008. GKK method finds IBP relations of the
integrals without doubled propagator, so only a small portion of loop
integrals need to be considered. In practice, such IBP relations are
found by the careful choice of IBP generating vectors $v_i^\mu$ in (\ref{IBP}), via Syzygy
computation \cite{Gluza:2010ws}. Several two-loop diagrams' IBP relations are given by this
method. Furthermore, the syzygy computation can be simplified by
linear algebra techniques \cite{Schabinger:2011dz}. However, GKK method does not indicate the geometric meaning of such
IBP-generating vectors. It is an interesting question to ask if these
vectors have any particular meaning in the loop-momentum space. 

In our paper, we illustrate the geometric meaning of the IBP
generating vectors for integral without doubled propagator. We
reformulate such a vector as a differential form by Poincar\'{e} dual. 
\begin{equation}
      v_i^\mu  \Leftrightarrow \omega,
\end{equation}
where $\omega$ is a rank-$(DL-1)$ differential form. Then we
show that it is {\it locally} proportional to the differential form $\Omega = dD_1 \wedge \ldots \wedge d D_k$,
\begin{equation}
  \label{eq:78}
  \omega \mid_{\mathcal S} \ \propto \Omega  \mid_{\mathcal S},
\end{equation}
where $D_i$'s are the sets of all denominators of the Feynman
integral and $\mathcal S$ is the unitarity cut solution. Geometrically, $\omega$ is along the normal
direction of the unitarity-cut surface. 

Furthermore, we design a geometric method to generate IBP
identities without doubled propagator. We consider the {\it primary
decomposition} of the unitarity cut solutions, 
\begin{equation}
  \label{eq:62}
  \mathcal S=\bigcup_{i=1}^n \mathcal S_i.
\end{equation}
By solving congruence equations, we construct differential form
$\omega_i$'s which is nonzero and proportional to $\Omega$ in
$\mathcal S_i$,
but vanishes on other branches,
\begin{equation}
  \label{eq:76}
\left \{
  \begin{array}{c}
     \omega_i |_{\mathcal S_i} = \ (\alpha \wedge \Omega)
  |_{\mathcal S_i}  \\
 \omega_i |_{\mathcal S_j} = 0  |_{\mathcal S_j} ,\quad j \not= i
 \end{array}
\right . ,
\end{equation}
where $\alpha$ is an arbitrary non-zero $(DL-1-k)$-form. 
We use such $\omega_i$'s to generate the on-shell part of the IBP relations
without doubled propagator. Several two-loop four-point
and five-point examples are tested by our method.

This paper is organized as follows: in section \ref{IBP_diff}, we
reformulate IBP identities in terms of differential forms, and the
condition for IBP without doubled propagator is also reformulated. 
In section \ref{geometry}, we illustrate the geometric meaning of the
IBP-generating differential forms and present a new method for generating
the on-shell part of IBPs. In section \ref{examples}, several two-loop
examples based on our algorithm are given. 

\section{Integration-by-Parts identities in the formalism of
  differential form}
\label{IBP_diff}
We consider the $L$-loop Feynman integral,
\begin{equation}
  \label{integral}
I_{\{a_1, \ldots a_k\}}[N]=  \int \frac{d^D l_1}{i \pi^{D/2}} \ldots \frac{d^D l_L}{i \pi^{D/2}} \frac{N}{D_1^{a_1}
    \ldots D_k^{a_k}}.
\end{equation}
where $N$ is a polynomial in loop momenta. The integrand reduction and
unitarity solution structure has been studied by algebraic geometry
methods \cite{Zhang:2012ce, Mastrolia:2012an}. In the following
discussion, we will frequently use these algebraic geometry
methods. The mathematical notations are summarized in the Appendix and
the algebraic geometry reference is \cite{MR0463157}.



We find that it is convenient to rewrite IBP relations (\ref{IBP}) in terms of
differential forms. By Poincar\'{e} dual, the $(D\cdot L)$-dimensional vector $v_i^\mu$ is dual to a $D
\cdot L-1$ differential form $\omega$. Explicitly, 
\begin{equation}
  \label{Poincare_duality}
  \omega_{i_1 \ldots i_{(DL-1)}} \equiv \epsilon_{i_1 \ldots i_{(DL-1)}
    i_{DL}} v^{i_{DL}}  ,
\end{equation}
where $\epsilon_{i_1 \ldots i_{(DL-1)}
    i_{DL}}$ is the Levi-Civita symbol. In most of the following discussion, we
use the notations of differential forms, since it is convenient to
write down the exterior derivative and wedge products.  
We call a differential form {\it polynomial-valued}, if all the
components are polynomials in loop momenta, in the momentum-coordinate
basis. Note that this definition is consistent with linear
transformation of loop momenta.

The total derivative in (\ref{IBP}) can be dually written as, 
\begin{equation}
  \frac{\partial }{\partial l_i^\mu}\bigg(\frac{v_i^\mu}{D_1^{a_1}
    \ldots D_k^{a_k}} \bigg) \Leftrightarrow d\bigg(\frac{\omega}{D_1^{a_1}
    \ldots D_k^{a_k}}\bigg).
\label{eq:3}
  \end{equation}
So the IBP relation is 
\begin{equation}
  \label{IBP_dual}
  \int \frac{d\omega}{D_1^{a_1}
    \ldots D_k^{a_k}} -  \sum_{i=1}^k a_i\int \frac{dD_i\wedge \omega}{D_1^{a_1}
    \ldots D_i^{a_i+1} \ldots D_k^{a_k}}=0.
\end{equation}

Different choices of $v_i^\mu$, or $\omega$ lead
to different IBPs. One particularly interesting class of IBPs is {\it
  IBPs without doubled propagator}, which is described in the next subsection.

\subsection{IBPs without doubled propagator}
For a Feynman integral from Feynman rules, the powers of the
denominators $D_1,\ldots D_k$ in (\ref{integral}) are usually one or zero,  i.e., $a_i=0,1$, $i=1, \ldots k$.
We call such an integral, {\it integral without doubled propagator}.  We are interested in {\it IBPs without doubled propagators}, which is an IBP
whose teams are integrals without doubled propagator. 

We make an ansatz for an IBP without doubled propagator,
\begin{equation}
  \label{Ansatz}
  \int d\bigg(\frac{\omega}{D_1
    \ldots D_k}\bigg)=0,
\end{equation}
where $\omega$ is a polynomial-valued $(DL-1)$-form. Usually, the expansion of (\ref{IBP}) contains integrals with double
propagators, because,
\begin{equation}
  \label{eq:4}
  d\bigg(\frac{1}{D_i}\bigg) =-\frac{dD_i}{D_i^2}.
\end{equation}
However, a particular choice of $\omega$ can remove the
double power if,
\begin{equation}
  \label{IBP_without_double_propagator}
  dD_i \wedge \omega=f_i D_i dl_1^0\wedge \ldots \wedge d
  l_L^{D-1},\quad 
  i=1,\ldots j
\end{equation}
where $f_i$ is a polynomial.

\subsection{On-shell part of IBPs}
Sometimes we only focus on Feynman diagrams without pinched legs,
i.e., $a_i\geq 1, i=1,\ldots k$. We call the corresponding integrals {\it leading
  integrals}. On the other hand, we call integrals with at least one
$a_i<1$ {\it simpler integrals}. If we only keep the leading
integrals in an IBP relation, then the resulting formula
\begin{equation}
  \label{eq:6}
  \sum_{i} c_i I_{a_{i,1},\ldots a_{i,k}}[N_i]+ \ldots =0, 
\end{equation}
is called an {\it on-shell IBP relation}. $a_{i,j}>0, \forall i,j $.  Here ``$\ldots$'' denotes the
{\it simpler integrals}, and $N_i$'s are polynomial numerators. 

In this paper, we consider the on-shell IBP without double
propagators, namely,
\begin{equation}
  \label{IBP1}
  \sum_{i} c_i I_{1,\ldots 1}[N_i]+ \ldots =0,
\end{equation}

For the ansatz (\ref{Ansatz}) to generate an on-shell IBP without
doubled propagator, it is sufficient that,
\begin{equation}
  \label{on_shell_IBP_without_double_propagator}
  dD_i \wedge \omega=\sum_j f_{ij} D_j dl_1^0\wedge \ldots \wedge d l_L^{D-1},\quad 
  i=1,\ldots j
\end{equation}
where each $f_{ij}$ is a polynomial.  $\omega$ generates the IBP,
\begin{eqnarray}
  \label{eq:9}
 0= \int d\big( \frac{\omega}{D_1 \ldots D_k} \big) = \int \frac{d\omega}{D_1
    \ldots D_k}-\sum_{i=1}^k \sum_{j=1}^k \int \frac{ f_{ij} D_j
    dl_1^0\wedge \ldots \wedge d l_L^{D-1}}{D_1 \ldots D_i^2 \ldots
    D_k} ,
\end{eqnarray}
Pick up the on-shell part, we have 
\begin{equation}
\label{on_shell}
   0= \int \frac{d\omega}{D_1
    \ldots D_k}-\sum_{i=1}^k  \int \frac{ f_{ii}  dl_1^0\wedge \ldots
    \wedge d l_L^{D-1}}{D_1 \ldots D_k} + ... ,
\end{equation}
where $\ldots$ stands for simpler integrals. Note that this condition
(\ref{on_shell_IBP_without_double_propagator}) is weaker than the
condition (\ref{IBP_without_double_propagator}).

Furthermore, from (\ref{on_shell}), we have the following lemma,
\begin{lemma}
 If if all components of $\omega$ are
in the ideal $I=\langle D_1, \ldots D_k \rangle$, then it generates an IBP identity whose on-shell part is trivial.
 \end{lemma}
 \begin{proof}
   Let $\omega'=\sum_{i=1}^m w_i dx_1\wedge \ldots
   \wedge\hat{dx_i}\wedge \ldots \wedge dx_m$, where $m=LD$ and $\{x_1,\ldots
     x_m\}$ denote the loop momenta $\{l_1^0 ,\ldots
     l_L^{D-1}\}$. Suppose that every $w_i$ is in $I$, i.e., $w_i=\sum_{j=1}^k
     g_{ij} D_j$. Hence,
     \begin{eqnarray}
       \label{eq:18}
    0 &=&  \int d\big(\frac{\omega}{D_1 \ldots D_k}\big)=\sum_{i=1}^m \sum_{j=1}^k\int
    d\big(\frac{ g_{ij} D_j dx_1\wedge \ldots
   \wedge\hat{dx_i}\wedge \ldots \wedge dx_m}{D_1 \ldots
   D_k}\big)\nonumber \\
&=& \sum_{i=1}^m \sum_{j=1}^k\int
    d\big(\frac{ g_{ij}  dx_1\wedge \ldots
   \wedge\hat{dx_i}\wedge \ldots \wedge dx_m}{D_1 \ldots
   \hat{D_j}\ldots D_k}\big) .
     \end{eqnarray}
From the expansion of the expression, it is clear that each term
misses one of the denominators.  Therefore, $\omega'$ generates the
IBP,
\begin{equation}
  \label{eq:13}
  0=0 + \ldots ,
\end{equation}
where $\ldots$ stands for simpler integrals. The on-shell part is trivial.
 \end{proof}
From this lemma, if two rank-$DL-1$ forms $\omega_1$ and $\omega$
differ by such an $\omega'$, then $\omega_1$ and $\omega_2$ generate
the same on-shell IBP. 
If an $\omega$
satisfying (\ref{on_shell_IBP_without_double_propagator}), then 
$f \omega$ also
satisfies (\ref{on_shell_IBP_without_double_propagator}). Here $f$ is
a polynomial in loop momenta. So we can obtain more IBPs without
doubled propagator, by multiplying various $f$'s. Note that by Lemma 1, only when $f$ is a
polynomial in {\it irreducible scalar products}, the resulting $f
\omega$ generates a non-trivial on-shell IBP.

\section{A method to construct on-shell IBPs without doubled propagator}
\label{geometry}

We reformulate  (\ref{on_shell_IBP_without_double_propagator}) from the
viewpoint of algebraic geometry, and then illustrate how to find the
solution to  (\ref{on_shell_IBP_without_double_propagator}) with
computational algebraic geometry method.

\subsection{A condition for on-shell IBPs without doubled propagator}
With the background of algebraic geometry, we can reformulate the
condition (\ref{on_shell_IBP_without_double_propagator}) as the
differential geometry constraint in 
Proposition \ref{lemma_product}.
\begin{proposition}
\label{vanishing_lemma}
For an $\omega$ in (\ref{Ansatz}) to generate an on shell IBP without doubled propagator,
it is necessary that for each point on the cut solution, at the
corresponding cotangent space,
\begin{equation}
  \label{vanishing_id}
  (dD_i \wedge \omega)|_P=0, \quad \forall P \in \mathcal Z(I).
\end{equation}
If the ideal generate by the denominators is radical, then this
condition is also sufficient.
\end{proposition}
\begin{proof}
 By the definition, all $D_i$ vanish on $\mathcal S=\mathcal
 Z(I)$. So $\forall P\in Z(I)$, $(dD_i \wedge \omega)|_P=0$. 
On the other hand, 
\begin{equation}
  \label{eq:8}
   (dD_i \wedge \omega)=F_i dl_1^0\wedge \ldots \wedge d l_L^{D-1},\quad 
  i=1,\ldots k
\end{equation}
where each $F_i$ is a polynomial. (\ref{vanishing_id}) means that
$F_i$ vanish everywhere on $\mathcal S$. So by Hilbert's Nullstenllensatz,
$F_i\in \sqrt I$. If $I$ is radical, then $F_i\in I$ and so $F_i=\sum_j f_{ij} D_j$.
\end{proof}

To get some insights of (\ref{vanishing_id}), we consider the
cotangent space at $P$. We consider general case, 
for which the cut equation system is non-degenerate, i.e., 
\begin{equation}
  \label{eq:10}
  \dim \mathcal S_i=DL-k, \quad i=1, \ldots n
\end{equation}
where $k$ is the number of denominators. If $P$ is a {\it non-singlar
  point}, i.e., the Jacobian
\begin{equation}
  \label{eq:11}
  J=\det \bigg(\frac{\partial D_i}{\partial x_j}\bigg)|_P .
\end{equation}
has the rank $k$, then it is clearly that 
\begin{equation}
  \label{eq:12}
  (d D_1 \wedge \ldots \wedge d D_k)|_P \not=0.
\end{equation}
Therefore we have the following proposition,
\begin{proposition}
\label{lemma_product}
If $k\leq DL-1$ and all cut solutions have the dimension $DL-k$, for an $\omega$ in (\ref{Ansatz}) to generate an on shell IBP without doubled propagator,
it is necessary that for each non-singular point $P$ on the cut solution,
at the cotangent space,
\begin{equation}
  \omega|_P=(\alpha \wedge D_1 \wedge \ldots \wedge D_k )|_P .
\end{equation}
where $\alpha$ is a $(DL-k-1)$ form.
\end{proposition}
\begin{proof}
  Since at the non-singular point $P$, the Jacobian is non-zero. So
  locally we can choose a coordinator system, $(y_1,\ldots y_{DL})$
  such that,
  \begin{equation}
    \label{eq:14}
    y_1 = D_1, \quad \ldots,\quad y_k=D_k.
  \end{equation}
Expand $\omega|_P$ in this coordinator system. If $\omega|_P$ contains a
component proportional to $dy_1 \wedge \ldots \hat{dy_i}\ldots \wedge
dy_n$ and $i\leq k$, then 
\begin{equation}
  \label{eq:16}
  (dD_i \wedge \omega)|_P\not =0 .
\end{equation}
This is a violation to Proposition \ref{vanishing_lemma}. Collecting all
terms  proportional to $dy_1 \wedge \ldots \hat{dy_i}\ldots \wedge
dy_n$ and $i> k$, this lemma is clear. 
\end{proof}

Generically, the singular points on $\mathcal S$ only form a subset with lower
dimension. So for ``almost all points'' on $\mathcal S$, $\omega$ is
proportional to $dD_1 \wedge \ldots \wedge dD_k$. We may have an 
explicit ansatz,
\begin{equation}
  \label{attempting_ansatz}
 \omega= \alpha\wedge dD_1 \wedge \ldots dD_k.
\end{equation}
Here $\alpha$ is a polynomial-valued differential form. This indeed generates an on-shell IBP relation without double
propagator. However, this form may not generate enough IBP relations,
since proposition \ref{vanishing_lemma} is only a local condition while
(\ref{attempting_ansatz}) has a global expression.

We may generalize (\ref{attempting_ansatz}) as: a polynomial-valued
differential form $\omega$ which {\it locally} has the form,
\begin{equation}
  \label{ansatz}
  \omega|_{\mathcal S_i} = \alpha_i \wedge dD_1 \wedge \ldots dD_k .
\end{equation}
on each branch $\mathcal S_i$. $\alpha_i$'s are different polynomial $(DL-k-1)$-froms on
different branches. Then there are two questions,
\begin{itemize}
\item Given a set of $\alpha_i$'s, does such a polynomial-valued $\omega$ exist?
\item Given a set of $\alpha_i$'s, is there an algorithm to find such
  an $\omega$?
\end{itemize}
These questions will be answered in the next section, explicitly in
Theorem \ref{congruence}, by solving {\it congruence equations}.

\subsection{Local form and congruence equations}
To study the behaviour of a differential form near the cut, we
use the tool of Gr\"obner basis and polynomial divisions. Recall that $I$ has the primary decomposition $I=I_1\cap ... \cap I_n$. Let $G(I)$ be the Gr\"obner basis of $I$, and $G(I_i)$
be the Gr\"obner basis of $I_i$. We denote the equivalent classes
$[\ ]$ and $[\ ]_i$
as,
\begin{eqnarray}
\ [f]&=&[g], \quad \text{if } f-g \in I, \\
\ [f]_i&=&[g]_i, \quad \text{if } f-g \in I_i.
\end{eqnarray}
Intuitively, these equivalent classes characterise the limit of the
polynomials approaching the cut manifold. In practise, the unique
representative for $[f]$ (or $[f]_i$) can be chose as the remainder
of the polynomial division of $f$ over $G(I)$ (or $G(I_i)$).

Here we generalize the equivalent classes to polynomial-valued
differential forms. Two differential forms $\alpha$ and $\beta$ are in
the same equivalent classes, if and only if $\alpha$ and $\beta$ are
of the same rank and all polynomial components are in the same
equivalent classes. We still use $[\ ]$ and $[\ ]_i$ for differential
forms.

Then we rewrite the condition (\ref{ansatz}) as,
\begin{eqnarray}
  \label{eq:15}
  [\omega]_i= [\alpha_i \wedge dD_1 \wedge \ldots dD_k]_i .
\end{eqnarray}
For a large classes of diagrams, given an arbitrary set of
$\alpha_i$'s, such differential form $\omega$ exists. We have the
following theorem,

\begin{thm}
\label{congruence}
Let $I=\langle D_1, \ldots D_k \rangle$ be an ideal in the ring
$\mathbb C[x_1,\ldots x_{m}]$. $I=I_1 \cap \ldots I_n$ is its primary
decomposition and $J_i=\cap_{j=1}^i I_i$. Suppose that (1) for each component $\dim \mathcal
Z(I_i)=m-k$， (2) Each $(J_i+I_{i+1})$ is a
radical ideal, $i=1, \ldots n-1$. Then given an arbitrary set of
rank-$(m-k-1)$ polynomial-valued forms, $\alpha_i$, there exists a rank-$(m-1)$ form $\omega$
such that,
\begin{eqnarray}
  \label{eq:15}
  [\omega]_i= [\alpha_i \wedge dD_1 \wedge \ldots \wedge dD_k]_i .
\end{eqnarray}
\end{thm}
\begin{proof}
       We construct $\omega$ explicitly by solving congruence
       equations. Define $v_i=\alpha_i \wedge dD_1 \wedge \ldots
       \wedge dD_k$. First, the ideal $I_1+I_2$'s zero locus is
       $\mathcal Z(I_1+I_2)=\mathcal Z(I_1) \cap Z(I_2)$, which are all
       singular points on the algebraic set $\mathcal Z(I)$. Since
       $\dim \mathcal Z(I_i)=m-k$, the Jacobian $\partial
       D_i/\partial  x_j$'s rank is strictly less than $k$
       on $\mathcal Z(I_1+ I_2)$. In other words, $dD_1 \wedge \ldots
       \wedge  dD_k$ vanishes on $\mathcal Z(I_1+ I_2)$. Hence $v_1-v_2$
       vanishes on $\mathcal Z(I_1+ I_2)$. Then by using Hilbert
       Nullstenllensatz for each component and the condition that $I_1+I_2$
       is radical, $v_1-v_2$ is in $I_1+I_2$, i.e.,
       \begin{equation}
         \label{eq:17}
         v_1-v_2=a_1+a_2, \quad a_1\in I_1,\quad  a_2 \in I_2
       \end{equation}
       Define $v_{12}=v_1-a_1$. Then $[v_{12}]_1=[v_1]_1$ and
       $[v_{12}]_2=[v_2]_2$. Then by induction, we have a differential
       form $v_{1\ldots i}$ such that $[v_{1\ldots i}]_j=[v_j]_j$,
       $\forall 1\leq j\leq i$. The zero locus of $J_i+I_{i+1}$ is,
       \begin{equation}
         \label{eq:19}
         \mathcal Z(J_i + I_{i+1})=\bigcup_{j=1}^i \big(\mathcal Z(I_j) \cap
         \mathcal Z(I_{i+1}) ) .
       \end{equation}
       which are also singular points on the algebraic set $\mathcal
       Z(I)$. Since $[v_{1\ldots i}]_j=[\alpha_j \wedge dD_1 \wedge
       \ldots \wedge dD_k]_j$, $v_{1\ldots i}$  vanishes on $\mathcal Z(I_j) \cap
         \mathcal Z(I_{i+1})$. Hence both   $v_{1\ldots i}$ and
         $v_{i+1}$ vanish on $\mathcal Z(J_i + I_{i+1})$. Then by using
         Hilbert Nullstellensatz, we obtain the differential form
         $v_{1\ldots (i+1)}$. Finally we denote $v_{1\ldots n}=\omega$.
\end{proof}

A large classes of 4D high-loop diagrams satisfy two conditions in
the above proposition. So we can construct $\omega$ for the IBP
without doubled propagator. The proof itself provides the algorithm
for obtaining $\omega$. This algorithm is realized by our Mathematica
and Macaulay2 \cite{M2}
package, {\sc MathematicaM2}. \footnote{This package can be downloaded
  from \url{http://www.nbi.dk/~zhang/MathematicaM2.html}.}



\begin{remark}
\label{factorization}
Note that in practice, after obtaining the differential form $\omega$
which satisfies (\ref{vanishing_id}), there may exist further
simplification. The form $\omega$ may factorize as,
\begin{equation}
  \label{eq:61}
\omega=f \omega'  .
\end{equation}
where $f$ is a polynomial in loop momenta and $\omega'$ is a
polynomial-valued form. If $\omega$ satisfies
(\ref{vanishing_id}), there is no guarantee that $\omega'$ also
satisfies (\ref{vanishing_id}). However, if accidentally $\omega'$
satisfies (\ref{vanishing_id}), we can instead use $\omega'$ to
generate an IBP without doubled propagator. 
\end{remark}

\section{Examples}
\label{examples}

In this section, we demonstrate our method by several $4D$ two-loop
examples. In each case, we generate the $4D$ on-shell part of the IBP
identities by our differential geometry method, via local form and
congruence equations. To simplify the process, we combine integrand
reduction method and our differential geometry approach for IBP computations.

\subsection{Planar double box}
Consider the $4D$ planar double box with $4$ massless legs, $p_1$,
$p_2$, $p_3$ and $p_4$. 
\begin{figure}
\center
  \includegraphics[width=2.8in]{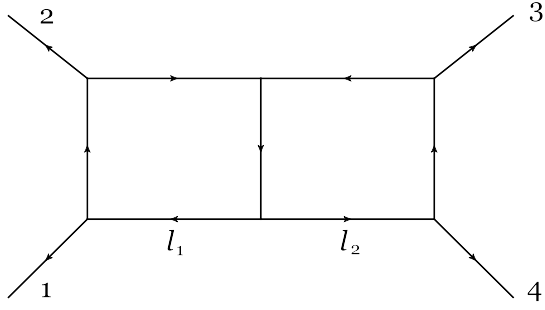}\\
  \caption{Planar double box with $4$ massless legs}\label{dbox}
\end{figure}
The two loop momenta are $l_1$ and
$l_2$. There are $7$ denominators for double box integrals,
\begin{gather}
  \label{eq:21}
  D_1=l_1^2,\quad D_2=(l_1-p_1)^2,\quad D_3=(l_1-p_1-p_2)^2, \nonumber\\
  D_4=(l_2-p_3-p_4)^2,\quad D_5=(l_2-p_4)^2, \quad D_6=l_2^2,\quad D_7=(l_1+l_2)^2.
\end{gather}

Instead of using Minkowski components of $l_1$ and $l_2$, we use van
Neerven-Vermaseren basis,
\begin{eqnarray}
  \label{eq:22}
  x_1&=&l_1 \cdot p_1,\quad  x_2=l_1 \cdot p_2, \quad  x_3=l_1 \cdot p_4
  ,\quad  x_4=l_1 \cdot \omega, \nonumber \\
 y_1&=&l_2 \cdot p_1,\quad  y_2=l_2 \cdot p_2, \quad  y_3=l_2 \cdot p_4
  ,\quad  y_4=l_2 \cdot \omega .
\end{eqnarray}
where $\omega$ is the vector which is perpendicular to all externel
legs and $\omega^2=t u/s$. The denominators have the parity symmetry,
\begin{equation}
  \label{dbox_parity}
  x_4 \leftrightarrow -x_4, \quad y_4 \leftrightarrow -y_4.
\end{equation}
Define the ideal $I\equiv \langle D_1,
\ldots D_7 \rangle$. The ISPs are $\{x_3,x_4,y_1,y_4\}$. Integrals
with numerators linear in $x_4$ or $y_4$ are spurious, i.e., vanish by
the orthogonal property of $\omega$. 

The $4D$ double box cut has $6$ branches, 
\begin{equation}
  \label{eq:26}
  I=I_1 \cap I_2 \cap I_3 \cap I_4 \cap I_5 \cap I_6,
\end{equation}
where,
\begin{gather}
  \label{eq:28}
 I_1= \langle x_1,-s-2 y_1-2 y_2,s-2 x_2,y_3,x_3,t-2 y_1+2 y_4,2
 x_4-t\rangle,
\\
I_2= \langle y_1,x_1,s+2 y_2,s-2 x_2,y_3,t+2 y_4,-t+2 x_3+2 x_4\rangle,
\\
I_3= \langle x_1,-s-2 y_1-2 y_2,s-2 x_2,y_3,x_3,-t+2 y_1+2 y_4,t+2 x_4\rangle,
\\
I_4= \langle y_1,x_1,s+2 y_2,s-2 x_2,y_3,2 y_4-t,t-2 x_3+2 x_4\rangle,
\\
I_5= \langle x_1,s+2 y_1+2 y_2,s-2 x_2,y_3,-s t+2 s x_3+2 s y_1+4 x_3 y_1,\nonumber\\ t-2 y_1+2 y_4,t-2 x_3+2 x_4\rangle,
\\
I_6= \langle x_1,s+2 y_1+2 y_2,s-2 x_2,y_3,-s t+2 s x_3+2 s y_1+4 x_3
  y_1,\nonumber\\-t+2 y_1+2 y_4,-t+2 x_3+2 x_4\rangle
\end{gather}
Note that under the parity symmetry (\ref{dbox_parity}), the primary
ideals are permuted,
\begin{equation}
  \label{eq:35}
  I_1 \leftrightarrow I_3,\quad I_2 \leftrightarrow I_4,\quad I_5
  \leftrightarrow I_6
\end{equation}

We can first carry out the integrand reduction for double-box
numerators. The irreducible numerator terms have the form,
\begin{equation}
  \label{eq:23}
  x_3^m y_1^n x_4^a y_4^b.
\end{equation}
The renormalizability condition requires that $0\leq m+a\leq 4, 0\leq n+b \leq 4,
0 \leq m+n+a+b \leq 6$. Furthermore, the Gr\"obner basis and
polynomial division method \footnote{The package for integrand reduction can be downloaded
  from \url{http://www.nbi.dk/~zhang/BasisDet.html}. } \cite{Zhang:2012ce} determines that
, the integrand basis
$\mathcal B=\mathcal B_1 \cup \mathcal B_2 $,
contains $32$ terms,
\begin{equation}
  \label{eq:24}
  \mathcal B_1 =\{x_3^4 y_1,x_3 y_1^4,x_3^4,x_3^3 y_1,x_3 y_1^3,y_1^4,x_3^3,x_3^2 y_1,x_3 y_1^2,y_1^3,x_3^2,x_3 y_1,y_1^2,x_3,y_1,1\}
\end{equation}
and
\begin{gather}
  \label{eq:25}
  \mathcal B_2 =\{x_4,x_3 x_4,x_3^2 x_4,x_3^3 x_4,x_4 y_1,y_4,x_3
    y_4,x_3^2 y_4,x_3^3 y_4,x_3^4 y_4,y_1 y_4,x_3 y_1 y_4,y_1^2
    y_4,\nonumber \\x_3 y_1^2 y_4,y_1^3
   y_4,x_3 y_1^3 y_4\}.
\end{gather}
Note all terms in $\mathcal B_2$ are spurious. So we focus on further
reducing the $16$ terms in $\mathcal B_1$ via IBPs. We divide our
algorithm in several steps,
\begin{enumerate}
\item Evaluate $\Omega=dD_1 \wedge \ldots \wedge dD_7$ and the local
  forms $[\Omega]_i$. Direct computation gives,
  \begin{gather}
    \label{eq:30}
 \Omega=  \frac{128 s}{t^3 (s+t)^3}\bigg( (s (x_4 (y_1+y_3)-y_4 (x_1+x_3))+t (y_4 (x_2-x_1)+x_4 (y_1-y_2)))
    \nonumber\\(s (y_1+y_3)+t (y_1+y_2+2 y_3)) dx_1\wedge dx_2\wedge
    dx_3\wedge dx_4\wedge dy_1\wedge dy_2\wedge dy_3\nonumber\\+s y_4
    (s (y_4 (x_1+x_3)-x_4 (y_1+y_3))+t (y_4 (x_1-x_2)+x_4
    (y_2-y_1)))\nonumber\\ dx_1\wedge dx_2\wedge dx_3\wedge dx_4\wedge
    dy_1\wedge dy_3\wedge dy_4 \nonumber\\+s y_4 (s (y_4 (x_1+x_3)-x_4
    (y_1+y_3))+t (y_4 (x_1-x_2)+x_4 (y_2-y_1))) \nonumber \\dx_1\wedge
    dx_2\wedge dx_3\wedge dx_4\wedge dy_2\wedge dy_3\wedge dy_4\nonumber\\- (s
    (y_4 (x_1+x_3)-x_4 (y_1+y_3))+t (y_4 (x_1+x_2+2 x_3)-x_4
    (y_1+y_2+2 y_3))) \nonumber\\(s (x_1+x_3)+t (x_1-x_2)) dx_1\wedge
    dx_2\wedge dx_3\wedge dy_1\wedge dy_2\wedge dy_3\wedge
    dy_4\nonumber\\-s x_4 (s (x_4 (y_1+y_3)-y_4 (x_1+x_3))+t (x_4
    (y_1+y_2+2 y_3)-y_4 (x_1+x_2+2 x_3)))\nonumber \\dx_1\wedge
    dx_2\wedge dx_4\wedge dy_1\wedge dy_2\wedge dy_3\wedge dy_4 \bigg).
  \end{gather}
The canonical representative of $[\Omega]_i$ is
  obtained by polynomial division.  For example, on the first branch,
  \begin{gather}
    \label{eq:29}
    [\Omega]_1=-\frac{64 s^2 y_1 (t-2 y_1)}{t^2 \
(s+t)^2} (dx_1\wedge dx_2\wedge dx_3\wedge \
dx_4\wedge dy_1\wedge dy_2\wedge dy_3\nonumber\\-dx_1\wedge dx_2\wedge \
dx_3\wedge dx_4\wedge dy_1\wedge dy_3\wedge dy_4-dx_1\wedge \
dx_2\wedge dx_3\wedge dx_4\wedge dy_2\wedge dy_3\wedge dy_4).
  \end{gather}
\item Verify that the two conditions in Theorem \ref{congruence} hold. In this
  case, $k=7$ and $m=DL=8$, so $m-k=1$. On the other hand, all six
  branches are one-dimensional. Furthermore, define $J_i=\cap_{j=1}^i
  I_i$. Directly commutative algebra computations indicate that
  $J_i+I_{i+1}$ is radical, for $i=1,2,3,4,5$.
\item Solve the congruence equations in the polynomial ring. Let
  $\eta_i$, $i=1,\ldots, 6$ be 7-forms satisfy the following equations,
  \begin{equation}
    \label{eq:31}
  \left\{ 
  \begin{array}{l l}
    \ [\eta_i]_j=[\Omega]_j & \quad j=i\\
    \ [\eta_i]_j=0 & \quad j\not=i, \quad j=1,\ldots, 6
  \end{array} \right.
  \end{equation}
The solution for $\eta_i$'s can be quickly obtained by our package
{\sc MathematicaM2}. For example,
\begin{gather}
  \label{eq:32}
\eta_1=
  -\frac{16 s (s (t (x_4+2 y_1+y_4)-2 (x_3 (2 y_1+y_4)+y_1 (x_4+2
    (y_1+y_4))))-8 x_3 y_1 (y_1+y_4)) }{t^2 (s+t)^2} \nonumber
  \\(dx_1\wedge dx_2\wedge dx_3\wedge dx_4\wedge dy_1\wedge dy_2\wedge
  dy_3-dx_1\wedge dx_2\wedge dx_3\wedge dx_4\wedge dy_1\wedge
  dy_3\wedge dy_4\nonumber \\-dx_1\wedge dx_2\wedge dx_3\wedge
  dx_4\wedge dy_2\wedge dy_3\wedge dy_4) .
\end{gather}
It is easy to check that,
\begin{equation}
  \label{eq:33}
 [\eta_1]_1=[\Omega]_1, \quad  [\eta_1]_2=[\eta_1]_3=[\eta_1]_4=[\eta_1]_5=[\eta_1]_6=0.
\end{equation}
\item Find all the IBP relations generated by $f \eta_j$ according to (\ref{on_shell}), where $f\in B$
  is a term from the integrand basis. For $4D$ double box case, the
  process can be sped up by using the parity symmetry. Define the
  7-forms according to the permutation of primary ideals,
  \begin{equation}
    \label{eq:36}
    v_1=\eta_1+\eta_3,\quad v_2=\eta_2+\eta_4,\quad v_3=\eta_5+\eta_6 
  \end{equation}
Then $v_i$'s, $i=1,2,3$ are even under the parity symmetry. Hence, we
can consider  IBP relations generated by $f v_j$, where $f\in
B_1$. In this way, we avoid the redundancy from spurious terms. 
For
example, explicitly,
\begin{gather}
  \label{eq:37}
v_1=\frac{32s}{t^2(s+t)^2}\bigg(-(s (t (x_4+y_4)-2 (x_3 y_4+x_4 y_1+2 y_1 y_4))-8 x_3 y_1 y_4)\nonumber\\
dx_1\wedge dx_2\wedge dx_3\wedge dx_4\wedge dy_1\wedge dy_2\wedge
dy_3-2 y_1 (s (2 (x_3+y_1)-t)+4 x_3 y_1) \nonumber\\dx_1\wedge
dx_2\wedge dx_3\wedge dx_4\wedge dy_2\wedge dy_3\wedge
dy_4-2 y_1 (s (2 (x_3+y_1)-t)+4 x_3 y_1) \nonumber\\dx_1\wedge
dx_2\wedge dx_3\wedge dx_4\wedge dy_1\wedge dy_3\wedge dy_4 \bigg).
\end{gather}
Consider the form $w=y_1 v_1$. 
\begin{equation}
  \label{eq:27}
  dw=-\frac{32 s y_1 \left(s \left(-5 t+10 x_3+16 y_1\right)+32 x_3
      y_1\right)}{t^2 (s+t)^2} \mathbf m
\end{equation}
Here $\mathbf m$ is the measure, $\mathbf m=dx_1 \wedge
dx_2 \wedge dx_3 \wedge dx_4 \wedge dy_1 \wedge dy_2 \wedge dy_3
\wedge dy_4$. Furthermore, it is clear that $dD_i \wedge \omega=
f_{ij} D_j \mathbf m$. The related components are,
\begin{gather}
  \label{eq:34}
  f_{11}=0, \quad f_{22}=0, \quad f_{33}=0,\\
f_{44}=\frac{16 s y_1 \left(s t^2-2 s t x_3-6 s t y_1-4 s x_3 y_1+8 s y_1^2-16 t x_3 y_1+16 x_3 y_1^2\right)}{t^2 (s+t)^3}\\
f_{55}=\frac{16 s y_1 }{t^3 (s+t)^3} \big(s^2 t^2-2 s^2 t x_3-6 s^2 t y_1-4 s^2 x_3 y_1-8 s^2 y_1^2-16 s t x_3 y_1\nonumber\\-16 s t y_1^2-16 s x_3 y_1^2-32 t x_3
   y_1^2\big)\\
f_{66}=\frac{16 s y_1 \left(s t^2-6 s t x_3-6 s t y_1+4 s x_3 y_1+8 s
    y_1^2-16 t x_3 y_1+16 x_3 y_1^2\right)}{t^3 (s+t)^2}\\
f_{77}=\frac{64 s y_1 \left(s t-s x_3-3 s y_1-4 x_3 y_1\right)}{t^2 (s+t)^2}
\end{gather}
Using (\ref{on_shell}), we get one IBP relation,
\begin{eqnarray}
  \label{eq:40}
  -4 I_{\text{dbox}}[(l_1\cdot p_4)(l_2\cdot p_1)^2] -2 s
  I_{\text{dbox}}[(l_1\cdot p_4)(l_2\cdot p_1)] \nonumber \\-2 s
  I_{\text{dbox}}[(l_2\cdot p_1)^2]+s t I_{\text{dbox}}[(l_2\cdot p_1)]+\ldots=0
\end{eqnarray}
\end{enumerate}
Using this algorithm, we find that both $v_1$ and $v_2$ provide $3$
IBP relations, while $v_3$ provides $6$ IBP relations. These relations
are linearly independent. So our method reduces the number of double box integrals
from $16$ to $16-12=4$. The resulting $4$ integrals can be chosen as 
\begin{equation}
  \label{eq:38}
  I_{\text{dbox}}[1],\quad I_{\text{dbox}}[l_1 \cdot p_4], \quad I_{\text{dbox}}[l_2 \cdot
  p_1], \quad I_{\text{dbox}}[(l_1\cdot p_4)(l_2 \cdot
  p_1)]
\end{equation}
Furthermore, the symmetry of double box determines that,
\begin{equation}
  \label{eq:39}
  I_{\text{dbox}}[l_1 \cdot p_4]= I_{\text{dbox}}[l_2 \cdot
  p_1].
\end{equation}
So we reduce the number of independent integrals to $3$. Our $4D$
formalism misses one IBP relation which can be obtained from the
$D$-dimensional formalism,
\begin{equation}
  \label{dbox_missing}
  I_{\text{dbox}}[(l_1 \cdot p_4)(l_2 \cdot p_1)]=\frac{1}{8} s t
  I_{\text{dbox}}[1] -\frac{3}{4} s I_{\text{dbox}}[l_1\cdot p_4]
  + \ldots .
\end{equation}
This identity occurs in the $O(\epsilon)$-order in a $D$-dimensional IBP
relation. So it cannot be detected by the pure $4D$ IBP
formalism. Including this missing IBP, all integrals for $4D$ double
box are reduced to two master integrals,
\begin{equation}
  \label{eq:41}
  I_{\text{dbox}}[1],\quad I_{\text{dbox}}[l_1 \cdot p_4],
\end{equation}
and we verified that the result is consistent with the $4D$ limit of
the output of \FIRE{}. For example,
\begin{eqnarray}
  \label{eq:42}
  I_{\text{dbox}}[(l_1 \cdot p_4)^2]&=&\frac{t}{2} I_{\text{dbox}}[l_1
  \cdot p_4] + \ldots, \\
  I_{\text{dbox}}[(l_1 \cdot p_4)^3]&=&\frac{t^2}{4} I_{\text{dbox}}[l_1
  \cdot p_4]+ \ldots, \\
 I_{\text{dbox}}[(l_1 \cdot p_4)^4]&=&\frac{t^3}{8} I_{\text{dbox}}[l_1
  \cdot p_4]+ \ldots , \\
 I_{\text{dbox}}[(l_1 \cdot p_4)^2 (l_2 \cdot p_1)]&=&-\frac{s^2
   t}{16} I_{\text{dbox}}[1]+\frac{3s^2}{8} I_{\text{dbox}}[l_1\cdot
 p_4] +
 \ldots ,\\
 I_{\text{dbox}}[(l_1 \cdot p_4)^3 (l_2 \cdot p_1)]&=&\frac{s^3
   t}{32} I_{\text{dbox}}[1]-\frac{3s^3}{16} I_{\text{dbox}}[l_1\cdot p_4]+
 \ldots . 
\end{eqnarray}

\subsubsection{Comparison with GKK method}
It is interesting to see the relation between our method and GKK
method \cite{Gluza:2010ws}. GKK method solves syzygy equations for generating vectors
without doubled propagator. We treat the generating vector $v$ as a dual differential form
$\omega$. On
each branch it is easy to find the local form of $\omega$ and finally
we combine
local forms together by solving congruence equations. So far, our
method is limited to $4D$ and the on-shell part.  

We compare the $4D$ and the on-shell part of the generating vectors
for double box from GKK method. There are three such vectors in
\cite{Gluza:2010ws} for double box with four massless legs, namely 
\begin{equation}
  \label{eq:43}
  v^{(1)}_{\text{GKK}}, \quad  v^{(2)}_{\text{GKK}},\quad  v^{(3)}_{\text{GKK}}
\end{equation}
To compare these with our result, we take the Poincar\'{e} dual of these
vectors, namely $\omega^{(1)}_{\text{GKK}}$, $\omega^{(2)}_{\text{GKK}}$
and $\omega^{(3)}_{\text{GKK}}$. Then we can verify that the on-shell
part is related to our result as,
\begin{eqnarray}
  \label{eq:44}
 \  [\omega^{(1)}_{\text{GKK}}]&=&\frac{t^2(s+t)^2}{64 s^2}
 \big ([\eta_1]+[\eta_2]+[\eta_3]+[\eta_4]-[\eta_5]-[\eta_6] \big) ,\\
\  [\omega^{(2)}_{\text{GKK}}]&=&\frac{t^2(s+t)^2}{64 s}
\big  (-[\eta_1]+[\eta_2]-[\eta_3]+[\eta_4]-[\eta_5]-[\eta_6] \big),\\
\  [\omega^{(3)}_{\text{GKK}}]&=&\frac{t^2(s+t)^2}{64 s}
\big  (\frac{s+2 (l_2\cdot k_1)}{s}[\eta_1]-[\eta_2]+\frac{s+2 (l_2\cdot
    k_1)}{s}[\eta_3]\\
&&-[\eta_4]-\frac{s+2 (l_2\cdot k_1)}{s}[\eta_5]-\frac{s+2 (l_2\cdot k_1)}{s}[\eta_6]\big).
\end{eqnarray}
So on-shell, $\omega^{(i)}_{\text{GKK}}$'s are the linear combination
of the differential form $\eta_i$'s. (The overall factor
$t^2(s+t)^2/(64s)$ comes from the normalization and has no significant
meaning.) The coefficients are the
same for branch pairs (under the parity symmetry), so the spurious
terms drop out in the IBP calculation.

Therefore, our method reproduces the $4D$ on-shell part of the double box result
from GKK. 

\subsection{Non-planar crossed box}
Our method also works for non-planar diagrams. For example, consider the $4D$ crossed box with $4$ massless legs, $p_1$,
$p_2$, $p_3$ and $p_4$. The two loop momenta are $l_1$ and
$l_2$. 
\begin{figure}
\center
  \includegraphics[width=2.8in]{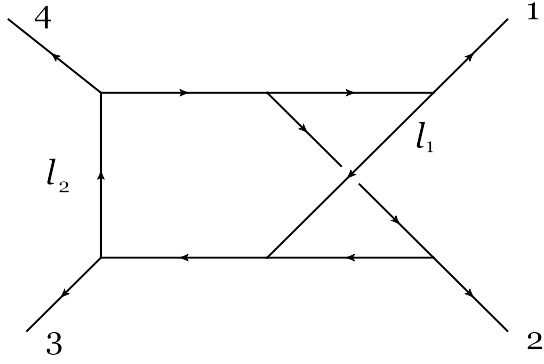}\\
  \caption{Non-planar double box with $4$ massless legs}\label{dbox}
\end{figure}

There are $7$ denominators for crossed box integrals,
\begin{gather}
  \label{eq:21}
  D_1=(l_1+p_1)^2,\quad D_2=l_1^2,\quad D_3=(l_2+p_3)^2, \nonumber\\
  D_4=l_2^2,\quad D_5=(l_2-p_4)^2, \quad D_6=(l_2-l_1+p_2+p_3)^2,\quad D_7=(l_2-l_1+p_3)^2.
\end{gather}
Again we use van
Neerven-Vermaseren basis,
\begin{eqnarray}
  \label{eq:22}
  x_1&=&l_1 \cdot p_1,\quad  x_2=l_1 \cdot p_2, \quad  x_3=l_1 \cdot p_3
  ,\quad  x_4=l_1 \cdot \omega, \nonumber \\
 y_1&=&l_2 \cdot p_1,\quad  y_2=l_2 \cdot p_2, \quad  y_3=l_2 \cdot p_3
  ,\quad  y_4=l_2 \cdot \omega .
\end{eqnarray}
where $\omega$ is the vector which is perpendicular to all externel
legs and $\omega^2=t u/s$. Again, the denominators have the parity symmetry,
\begin{equation}
  \label{xbox_parity}
  x_4 \leftrightarrow -x_4, \quad y_4 \leftrightarrow -y_4.
\end{equation}
Define the ideal $I\equiv \langle D_1,
\ldots D_7 \rangle$. The ISPs are $\{x_3,x_4,y_1,y_4\}$. Integrals
with numerators linear in $x_4$ or $y_4$ are spurious.

This diagram has the following symmetry,
\begin{gather}
  \label{xbox_symmetry}
  l_1 \to l_1-l_2+p_1+p_4, \quad l_2 \to -l_2,\\ 
p_1\to p_2, \quad p_2 \to p_1,\quad p_3 \to p_4,\quad p_4 \to p_3.
\end{gather}

The $4D$ crossed box cut has $8$ branches, 
\begin{equation}
  \label{eq:26}
  I=I_1 \cap I_2 \cap I_3 \cap I_4 \cap I_5 \cap I_6 \cap I_7 \cap I_8 ,
\end{equation}
where,
\begin{gather}
I_1=\langle -t+2 x_2-2 y_2,y_1+y_2,x_1,y_3,x_3+y_2,y_2+y_4,-\frac{t^2}{s}-\frac{2 t y_2}{s}-t+2 x_4\rangle,\\
I_2=\langle -t+2 x_2-2 y_2,y_1+y_2,x_1,y_3,x_3+y_2,y_4-y_2,\frac{t^2}{s}+\frac{2 t y_2}{s}+t+2 x_4\rangle,\\
I_3=\langle t+2 y_2,x_2,2 y_1-t,x_1,y_3,2 y_4-t,x_4-x_3\rangle,\\
I_4=\langle t+2 y_2,x_2,2 y_1-t,x_1,y_3,t+2 y_4,x_3+x_4\rangle,\\
I_5=\langle -t+2 x_2-2 y_2,y_1+y_2,x_1,y_3,x_3,y_2+y_4,\frac{t^2}{s}+y_2 (\frac{2 t}{s}+2)+t+2 x_4\rangle,\\
I_6=\langle -t+2 x_2-2 y_2,y_1+y_2,x_1,y_3,x_3,y_4-y_2,-\frac{t^2}{s}+y_2 (-\frac{2 t}{s}-2)-t+2 x_4\rangle,\\
I_7=\langle s +t+2 y_2,s+2 x_2,-s-t+2 y_1,x_1,y_3,-s-t+2 y_4,-s-t+2 x_3+2 x_4\rangle,\\
I_8=\langle s +t+2 y_2,s+2 x_2,-s-t+2 y_1,x_1,y_3,s+t+2 y_4,s+t-2 x_3+2 x_4\rangle,
\end{gather}
under the parity symmetry (\ref{xbox_parity}), the primary
ideals are permuted,
\begin{equation}
  \label{eq:35}
  I_1 \leftrightarrow I_2,\quad I_3 \leftrightarrow I_4,\quad I_5
  \leftrightarrow I_6 \quad  I_7,
  \leftrightarrow I_8.
\end{equation}

 The irreducible numerator terms have the form,
\begin{equation}
  \label{eq:23}
  x_3^m y_2^n x_4^a y_4^b.
\end{equation}
And the integrand reduction method \cite{Zhang:2012ce} determines that, the integrand
basis $\mathcal B=\mathcal B_1 \cup \mathcal B_2$, where
\begin{equation}
  \label{eq:46}
  \mathcal B_1=\{x_3 y_2^5,y_2^6,x_3^4 y_2,x_3 y_2^4,y_2^5,x_3^4,x_3^3 y_2,x_3 y_2^3,y_2^4,x_3^3,x_3^2 y_2,x_3 y_2^2,y_2^3,x_3^2,x_3
   y_2,y_2^2,x_3,y_2,1\},
\end{equation}
and 
\begin{gather}
  \label{eq:47}
  \mathcal B_2=\{x_4,x_3 x_4,x_3^2 x_4,x_3^3 x_4,x_4 y_2,y_4,x_3
  y_4,x_3^2 y_4,x_3^3 y_4,x_3^4 y_4,y_2 y_4,x_3 y_2 y_4,y_2^2 y_4,x_3
  y_2^2 y_4,y_2^3\nonumber \\ 
   y_4,x_3 y_2^3 y_4,y_2^4 y_4,x_3 y_2^4 y_4,y_2^5 y_4\}.
\end{gather}
There are $19$ terms in $\mathcal B_1$. 

Similarly, Define $\Omega=d D_1 \wedge \ldots d D_7$. By solving congruence equations, we obtain rank-$7$ forms
$\eta_i$, $i=1, \ldots 8$ such that,
\begin{eqnarray}
  \label{eq:48}
  [\eta_i]_j=\delta_{ij} [\Omega]_j,\quad 1\leq i,j \leq 8.
\end{eqnarray}
Again, to remove the spurious terms in $\mathcal B_2$, we define,
 \begin{equation}
    \label{eq:36}
    v_1=\eta_1+\eta_3,\quad v_2=\eta_2+\eta_4,\quad
    v_3=\eta_5+\eta_6,\quad v_4=\eta_7+\eta_8 .
  \end{equation}
We find that both $v_1$ and $v_3$ generate $4$ IBPs, while $v_2$ and
$v_4$ generate $3$ IBPs. Again these IBPs are linearly independent, so
our method generates $14$ relations. 

Furthermore, from the symmetry (\ref{xbox_symmetry}), we have,
\begin{eqnarray}
  \label{eq:50}
  2 I_{\text{xbox}}[l_1\cdot p_3] + I_{\text{xbox}}[l_2 \cdot p_2]&=&0+
  \ldots ,\\
   2 I_{\text{xbox}}[(l_1\cdot p_3)(l_2 \cdot p_2)] + I_{\text{xbox}}[(l_2 \cdot p_2)^2]&=&0+
  \ldots .
\end{eqnarray}
These $2$ relations are independent of the $14$ IBP relations we
obtained. Using these relations, we reduce the $19$ terms in $\mathcal
B_1$ to $3$ terms,
\begin{eqnarray}
  \label{eq:51}
  I_{\text{xbox}}[1],\quad I_{\text{xbox}}[l_1\cdot p_3], \quad I_{\text{xbox}}[(l_1\cdot p_3)(l_2\cdot p_2)] .
\end{eqnarray}

Again, there is one IBP relation missing in the pure $4D$
formalism. From FIRE \cite{FIRE}, we have,
\begin{eqnarray}
  \label{eq:52}
  I_{\text{xbox}}[(l_1\cdot p_3)(l_2\cdot p_2)] =\frac{1}{16} (t + s)
  t I_{\text{xbox}}[1] - \frac{3}{8} (s + 2 t)  I_{\text{xbox}}[l_1\cdot p_3] .
\end{eqnarray}
Combine $14+2+1=17$ relations together, we reduce the integrand terms
to two master integrals,
\begin{eqnarray}
  \label{eq:51}
  I_{\text{xbox}}[1],\quad I_{\text{xbox}}[l_1\cdot p_3]
\end{eqnarray}
For example,
\begin{eqnarray}
  \label{eq:53}
  I_{\text{xbox}}[(l_2\cdot p_2)^2]&=&-\frac{1}{8}  t(s+t) I_{\text{xbox}}[1]+\frac{3}{4} 
  (s+2t) I_{\text{xbox}}[l_1\cdot p_3] +\ldots, \\
 I_{\text{xbox}}[(l_1\cdot p_3)(l_2\cdot p_2)^2]&=&
\frac{-t(s^2+3 s t+2t^2)}{32} I_{\text{xbox}}[1]\nonumber \\ & & +\frac{
  (3s^2+8 s t+8 t^2)}{16} I_{\text{xbox}}[l_1\cdot p_3]  +\ldots,\\
 I_{\text{xbox}}[(l_2\cdot p_2)^3]&=&\frac{t(s^2+3 s t+2 t^2)}{16}
 I_{\text{xbox}}[1]  \nonumber \\ & &-\frac{(3 s^2+8 s t+8 t^2)}{8}
 I_{\text{xbox}}[l_1\cdot p_3] + ... 
\end{eqnarray}

\subsection{Slashed box}
Our method also works for diagram with less than $DL-1$ internal
lines. In these cases, the coefficients $\alpha$'s in (\ref{ansatz})
are not scalar functions, but differential forms. For example,  consider the $4D$ slashed box with $4$ massless legs, $p_1$,
$p_2$, $p_3$ and $p_4$. There are $5$ denominators for slashed box integrals,
\begin{gather}
  \label{eq:21}
  D_1=l_1^2,\quad D_2=(l_1-p_2)^2,\quad D_3=l_2^2, \quad D_4=(l_2-p_4)^2,\quad D_5=(l_1+l_2+p_1)^2,
\end{gather}
\begin{figure}
\center
  \includegraphics[width=2.2in]{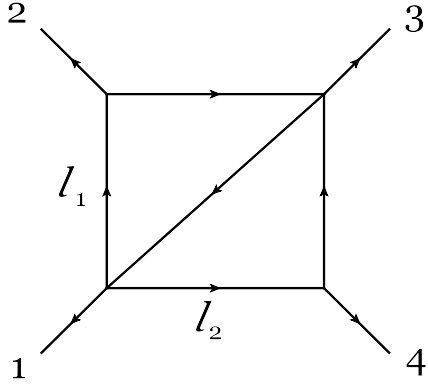}\\
  \caption{Planner slashed box with $4$ massless legs}\label{dbox}
\end{figure}
we use van
Neerven-Vermaseren basis,
\begin{eqnarray}
  \label{eq:22}
  x_1&=&l_1 \cdot p_1,\quad  x_2=l_1 \cdot p_2, \quad  x_3=l_1 \cdot p_4
  ,\quad  x_4=l_1 \cdot \omega, \nonumber \\
 y_1&=&l_2 \cdot p_1,\quad  y_2=l_2 \cdot p_2, \quad  y_3=l_2 \cdot p_4
  ,\quad  y_4=l_2 \cdot \omega .
\end{eqnarray}
where $\omega$ is the vector which is perpendicular to all externel
legs and $\omega^2=t u/s$. The denominators have the parity symmetry,
\begin{equation}
  \label{slashedbox_parity}
  x_4 \leftrightarrow -x_4, \quad y_4 \leftrightarrow -y_4.
\end{equation}
Define the ideal $I\equiv \langle D_1,
\ldots D_7 \rangle$. The ISPs are $\{x_1,x_3,x_4,y_1,y_2, y_4\}$. Integrals
with numerators linear in $x_4$ or $y_4$ are spurious.

The integrand basis for slashed box is $\mathcal B=\mathcal B_1 \cup
\mathcal B_2$ \cite{Zhang:2012ce},
\begin{gather}
  \label{eq:60}
  \mathcal B_1=\{x_3^3 y_2,x_3^3 y_1,x_3^2 y_2^2,x_1 x_3^2 y_2,x_1
  x_3^2 y_1,x_3^2 y_1^2,x_3 y_2^3,x_1 x_3 y_2^2,x_3 y_1 y_2^2,x_1^2
  x_3 y_2,x_1 x_3 y_1 y_2,x_3 y_1^2 y_2,\nonumber \\x_1^2 x_3 y_1,x_1 x_3
   y_1^2,x_3 y_1^3,x_1 y_2^3,x_1^2 y_2^2,x_1 y_1 y_2^2,x_1^3 y_2,x_1^2 y_1 y_2,x_1 y_1^2 y_2,x_1^3 y_1,x_1^2 y_1^2,x_1 y_1^3,x_3^3,x_3^2 y_2,x_1 x_3^2 \nonumber \\,x_3^2 y_1,x_3 y_2^2,x_1 x_3 y_2,x_3 y_1
   y_2,x_1^2 x_3,x_1 x_3 y_1,x_3 y_1^2,y_2^3,x_1 y_2^2,y_1 y_2^2,x_1^2
   y_2,x_1 y_1 y_2,y_1^2 y_2,x_1^3,x_1^2 y_1,x_1 y_1^2,\nonumber \\\ y_1^3,x_3^2,x_3 y_2,x_1 x_3,x_3 y_1,y_2^2,x_1 y_2,y_1 y_2,x_1^2,x_1
   y_1,y_1^2,x_3,y_2,x_1,y_1,1\},
\end{gather}
and
\begin{gather}
  \label{eq:63}
  \mathcal B_2=\{x_4,x_1 x_4,x_1^2 x_4,x_3 x_4,x_1 x_3 x_4,x_3^2
  x_4,x_4 y_1,x_1 x_4 y_1,x_1^2 x_4 y_1,x_3 x_4 y_1,x_1 x_3 x_4
  y_1,x_3^2 x_4 y_1, \nonumber \\x_4 y_1^2,x_1 x_4 y_1^2,x_4 y_1^3,x_4 y_2,x_1 x_4
   y_2, x_1^2 x_4 y_2,x_4 y_1 y_2,x_1 x_4 y_1 y_2,x_4 y_1^2
   y_2,y_4,x_1 y_4,x_1^2 y_4,x_1^3 y_4,x_3 y_4 \nonumber \\,x_1 x_3 y_4,x_1^2 x_3 y_4,x_3^2 y_4,x_1 x_3^2 y_4,x_3^3 y_4,y_1 y_4,x_1 y_1 y_4,x_1^2
   y_1 y_4,x_3 y_1 y_4,x_1 x_3 y_1 y_4,x_3^2 y_1 y_4,y_1^2 y_4,x_1 y_1^2 y_4,\nonumber \\x_3 y_1^2 y_4,y_2 y_4,x_1 y_2 y_4,x_1^2 y_2 y_4,x_3 y_2 y_4,x_1 x_3 y_2 y_4,x_3^2 y_2 y_4,y_1 y_2 y_4,x_1 y_1 y_2
   y_4,\nonumber \\x_3 y_1 y_2 y_4,y_2^2 y_4,x_1 y_2^2 y_4,x_3 y_2^2 y_4\}.
\end{gather}
There are $59$ terms in $\mathcal B_1$ and $52$ terms in $\mathcal
B_2$. Terms in $\mathcal B_2$ are all spurious.

This diagram has the following symmetry,
  \begin{gather}
  \label{slashed_box_symmetry}
  l_1 \to -l_2+p_4, \quad l_2 \to -l_1+p_2,\\ 
p_1\to p_3, \quad p_2 \to p_4,\quad p_3 \to p_1,\quad p_4 \to p_2.
\end{gather}

The $4D$ crossed box cut has $4$ branches, 
\begin{equation}
  \label{eq:26}
  I=I_1 \cap I_2 \cap I_3 \cap I_4 ,
\end{equation}
where,
\begin{eqnarray}
  \label{eq:45}
I_1&=& \{x_2,y_3,x_1 (-s-t)+y_1 (-s-t)+2 x_3 y_2,y_1
(-\frac{t}{s}-1)-\frac{t y_2}{s}+y_4,\nonumber \\
&&x_1 (-\frac{t}{s}-1)-x_3+x_4\},\\
I_2&=& \{x_2,y_3,x_1 (-s-t)+y_1 (-s-t)+2 x_3 y_2,y_1
(\frac{t}{s}+1)+\frac{t y_2}{s}+y_4,\nonumber \\
&&x_1 (\frac{t}{s}+1)+x_3+x_4\},\\
I_3&=&\{x_2,y_3,x_1 y_1 (\frac{2 t}{s}+2)+\frac{2 t x_1 y_2}{s}+t
x_1+t y_1+2 x_3 y_1,y_1 (\frac{t}{s}+1)+\frac{t y_2}{s}+y_4, \nonumber
\\ & & x_1 (-\frac{t}{s}-1)-x_3+x_4\},\\
I_4&=&\{x_2,y_3,x_1 y_1 (\frac{2 t}{s}+2)+\frac{2 t x_1 y_2}{s}+t
x_1+t y_1+2 x_3 y_1,y_1 (-\frac{t}{s}-1)-\frac{t y_2}{s}+y_4,
\nonumber \\
&&x_1 (\frac{t}{s}+1)+x_3+x_4\}
\end{eqnarray}
Under the parity symmetry, the ideals are permuted as,
\begin{equation}
  \label{eq:49}
  I_1 \leftrightarrow I_2, \quad I_3 \leftrightarrow I_4.
\end{equation}
We have $5$ denominators, so $\alpha_i$'s in (\ref{ansatz}) are rank-2
differential forms. We use a basis for all possible rank-2
differential form,
\begin{gather}
  \label{eq:55}
  \alpha^{(1)} = dx_1 \wedge dx_3,\quad \alpha^{(2)}=dx_1 \wedge d
  y_1,\quad \alpha^{(3)}=dx_1 \wedge dy_2, \quad \alpha^{(4)}=d x_3
  \wedge d y_1, \nonumber \\
\alpha^{(5)} =dx_3 \wedge dy_2,\quad \alpha^{(6)} =dy_1 \wedge dy_2,
\quad \alpha^{(7)}=dx_4 \wedge dy_4, \quad \alpha^{(8)}=dx_1 \wedge
dx_4 \nonumber \\
\alpha^{(9)}=dx_3 \wedge dx_4, \quad \alpha^{(10)}=dy_1 \wedge dx_4,
\quad \alpha^{(11)}=dy_2 \wedge dx_4, \quad \alpha^{(12)}=dx_1 \wedge
dy_4\nonumber \\
\alpha^{(13)}=dx_3 \wedge d y_4, \quad \alpha^{(14)}=dy_1 \wedge dy_4,
\quad
\alpha^{(15)}=dy_2 \wedge dy_4
\end{gather}
Note that all components in $dD_1 \wedge \ldots \wedge dD_5$ contains
$d x_2 \wedge d y_3$. So we do not list rank-$2$ forms containing $d x_2$
or $d y_3$. Now we define,
\begin{gather}
  \label{eq:57}
  \Omega^{(i)}=\alpha^{(i)} \wedge dD_1 \wedge \ldots \wedge dD_5,
  \quad 1 \leq i \leq 15
\end{gather}
Then we solve congruence equations to get $60$ 7-forms,
$\omega^{(i)}_j$, $1 \leq i \leq 15$, $1 \leq j \leq 4$, such that, 
\begin{eqnarray}
  \label{eq:64}
  [\omega^{(i)}_j]_k= \delta_{jk} [\Omega^{(i)}]_k.
\end{eqnarray}

We can use $\omega^{(i)}_j$'s to generate on-shell IBPs without
doubled propagator. Again, to remove spurious terms, we define
\begin{eqnarray}
  \label{eq:65}
  v_{2 i-1} &=& \omega^{(i)}_1 + \omega^{(i)}_2 \nonumber \\
  v_{2 i} &=& \omega^{(i)}_3+ \omega^{(i)}_4, \quad 1 \leq i \leq 15
\end{eqnarray}
Then all $v_i$'s are parity-even and we can use $f v_i$, $f \in
\mathcal B_1$, to generate IBP relations. 

However, the new feature for this diagram is that,
we can use Remark. \ref{factorization} to simplify the differential
form and get more IBPs. For example, 
\begin{equation}
  \label{eq:67}
  v_{13}=-\frac{16 \left(s \left(t \left(x_1+y_1\right)+2 \left(x_1+x_3\right) y_1\right)+2 t x_1 \left(y_1+y_2\right)\right)}{s^2 t^2 (s+t)}\tilde v_{13},
\end{equation}
where, 
\begin{gather}
  \label{eq:66}
  \tilde v_{13}=(s+t) (s+2 y_2) dx_1\wedge dx_2\wedge dx_3\wedge
  dx_4\wedge dy_1\wedge dy_3\wedge dy_4\nonumber\\+(s+t) (t+2 x_3)
  dx_1\wedge dx_2\wedge dx_4\wedge dy_1\wedge dy_2\wedge dy_3\wedge
  dy_4\nonumber\\+t (s+2 y_2) dx_1\wedge dx_2\wedge dx_3\wedge
  dx_4\wedge dy_2\wedge dy_3\wedge dy_4\nonumber\\-s (t+2 x_3)
  dx_2\wedge dx_3\wedge dx_4\wedge dy_1\wedge dy_2\wedge dy_3\wedge
  dy_4 .
\end{gather}
We can check that 
\begin{equation}
  \label{eq:68}
  [dD_i \wedge \tilde v_{13}]=0,\quad 1\leq i \leq 5 
\end{equation}
So instead, we can use $\tilde v_{13}$ to generate IBPs. In this
manner, we get more IBPs. Similarly, $v_{14}$ factorizes and we can
define a new rank-7 form $\tilde v_{14}$ for IBP generation. Other
$v_i$'s do not have non-trivial factorization. Using all $v_i$
($\tilde v_i$)'s , we get $51$ IBPs. 

Furthermore, $\Omega^{i}$ themselves also have the factorization
property. For example,
\begin{equation}
  \label{eq:69}
  \Omega^{(1)}=-\frac{32 x_4}{t^3 (s+t)^3} \tilde  \omega^{(1)},
\end{equation}
where,
\begin{gather}
  \label{eq:70}
  \tilde \Omega^{(1)}=-s (s+t) dx_1\wedge dx_2\wedge dx_3\wedge dx_4\wedge dy_1\wedge
  dy_3\wedge dy_4 \nonumber \\ (s (y_4 (t+x_1+x_3)-x_4 (y_1+y_3))+t (y_4
  (x_1+x_2)-x_4 (y_1+y_2))) \nonumber \\+s t dx_1\wedge dx_2\wedge
  dx_3\wedge dx_4\wedge dy_2\wedge dy_3\wedge dy_4 \nonumber \\(s (y_4
  (x_3-x_1)+x_4 (y_1-y_3))+t (x_4 (y_1+y_2)-y_4 (x_1+x_2))) \nonumber
  \\-t (s+t) (s (t (y_1-y_3)-2 x_1 y_3+2 x_3 y_1)+t (t (y_1+y_2)+2
  (x_3 (y_1+y_2)-y_3 (x_1+x_2)))) \nonumber \\ dx_1\wedge dx_2\wedge
  dx_3\wedge dx_4\wedge dy_1\wedge dy_2\wedge dy_3 .
\end{gather}
We can verify that,
\begin{equation}
  \label{eq:68}
  [dD_i \wedge \tilde  \Omega^{(1)}]=0,\quad 1\leq i \leq 5 
\end{equation}
So we can use $\tilde \Omega^{(1)}$ to generate IBPs. Similarly,
$\Omega^{(6)}$, $\Omega^{(8)}$, $\Omega^{(9)}$,
$\Omega^{(14)}$ and $\Omega^{(15)}$ also factorize. Using $\tilde
\Omega$ forms, we get $4$ more independent IBPs.

 Note that although $\Omega^{(1)}$ itself has the form $\alpha\wedge dD_1 \wedge \ldots \wedge
dD_5$, where $\alpha$ is a polynomial-valued differential
form. However, $\tilde \Omega^{(1)}$ cannot be expressed as a product
of 
polynomial-valued form and $dD_1 \wedge \ldots \wedge
dD_5$. So $\tilde \Omega^{(1)}$ does not satisfy the conditions in
Theorem. \ref{congruence} and there is no way to solve the
congruence equation,
\begin{gather}
  \label{eq:71}
  [\tilde \Omega^{(1)}_j]_k = \delta_{jk} [\tilde
  \Omega^{(1)}]_k,\quad 1\leq k \leq 4
\end{gather}
to get more differential forms. 

In summary, from differential forms, we get $51+4=55$ IBP
relations. Furthermore, using the symmetry condition
(\ref{slashed_box_symmetry}),
we have,
\begin{equation}
  \label{eq:72}
  I_{\text{slashed}}[l_2 \cdot p_1]=-I_{\text{slashed}}[l_1 \cdot
  p_3]+\frac{t}{2} I_{\text{slashed}}[1].
\end{equation}
So we have $59-55-1=3$ integrals left, 
\begin{equation}
  \label{eq:73}
  I_{\text{slashed}}[1],\quad I_{\text{slashed}}[l_1\cdot p_1] ,\quad I_{\text{slashed}}[(l_1\cdot p_1)^2]
\end{equation}
From FIRE \cite{FIRE}, there are two missing IBPs,
\begin{eqnarray}
  \label{eq:74}
  I_{\text{slashed}}[l_1\cdot p_1] = -\frac{s t}{2 u}
  I_{\text{slashed}}[1], \\
  I_{\text{slashed}}[(l_1\cdot p_1)^2] = \frac{s^2 t^2}{4 u^2}
  I_{\text{slashed}}[1] .
\end{eqnarray}
So the $59$ integrand terms reduce to $1$ master integral,
$I_{\text{slashed}}[1]$. For example,
\begin{eqnarray}
  \label{eq:75}
   I_{\text{slashed}}[l_1\cdot p_4] &= &-\frac{t}{2}
   I_{\text{slashed}}[1] ,\\
    I_{\text{slashed}}[(l_1\cdot p_1) (l_1 \cdot p_4)]&=&\frac{s
      t^2}{4 u}  I_{\text{slashed}}[1] ,\\
     I_{\text{slashed}}[(l_1\cdot p_1) (l_2 \cdot p_1)]&=&\frac{s^2
      t^2}{2 u^2}  I_{\text{slashed}}[1] ,\\
    I_{\text{slashed}}[(l_2\cdot p_1) (l_2 \cdot p_2)]&=&\frac{s^2
      t}{4 u}  I_{\text{slashed}}[1] .
\end{eqnarray}

\subsection{Turtle box}
Now consider the  $4D$ two-loop turtle box with $5$ massless legs, $p_1$,
$p_2$, $p_3$, $p_4$ and $p_5$.  This system is considerably more
difficult than the $4$-point two-loop cases, since the kinematics is complicated.
\begin{figure}
\center
  \includegraphics[width=2.8in]{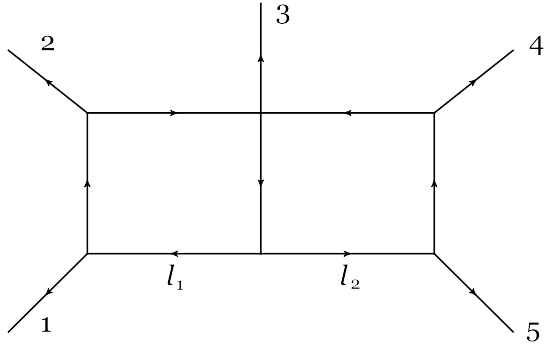}\\
  \caption{Planar double box with $5$ massless legs}\label{dbox}
\end{figure}
The two loop momenta are $l_1$ and
$l_2$. There are $7$ denominators for crossed box integrals,
\begin{gather}
  \label{eq:21}
 D_1=l_1^2,\quad D_2=(l_1-p_1)^2,\quad D_3=(l_1-p_1-p_2)^2, \nonumber\\
  D_4=(l_2-p_5)^2,\quad D_5=(l_2-p_4-p_5)^2, \quad D_6=l_2^2,\quad D_7=(l_1+l_2)^2.
\end{gather}

In this case, we find that it is easier to calculate differential
forms and IBP identity in spinor helicity formalism, and then convert
the result to   van
Neerven-Vermaseren basis in the final step. Define,
\begin{eqnarray}
  \label{eq:54}
  l_1^\mu&=&\alpha_1 p_1^\mu+\alpha_2 p_2^\mu+ \frac{s_{12} \alpha_3}{\langle 14 \rangle [42]} 
  \frac{[1|\gamma^\mu|2\rangle}{2} + \frac{s_{12} \alpha_4}{\langle 24 \rangle [41]} 
  \frac{[2|\gamma^\mu|1\rangle}{2}, \\
 l_2^\mu&=&\beta_1 p_4^\mu+\beta_2 p_5^\mu+ \frac{s_{12} \beta_3}{\langle 41 \rangle [15]} 
  \frac{[4|\gamma^\mu|5\rangle}{2} + \frac{s_{12} \beta_4}{\langle 51 \rangle [14]} 
  \frac{[5|\gamma^\mu|4\rangle}{2}.
\end{eqnarray}
Furthermore, to simplify the computation, we use {\it
  momentum-twistor} variables \cite{Hodges:2009hk, Mason:2009qx}  for $s_{ij}$, $\langle i,j\rangle $ and
$[i,j]$. The advantage is that all constraints like momentum
conservation and Schouten identities are resolved in {\it
  momentum-twistor} variables.

The ISPs are 
\begin{eqnarray}
  \label{eq:22}
  a=l_1 \cdot p_4,\quad  b=l_1 \cdot p_5, \quad  c=l_2 \cdot p_1
  ,\quad  d=l_2 \cdot p_2, 
\end{eqnarray}
The integrand basis contains $32$ terms,
\begin{gather}
  \label{eq:56}
  \mathcal B=\{b^4 c,b^4 d,b c d^3,b d^4,a b^3,b^4,b^3 c,b^3 d,b c d^2,b d^3,c
  d^3,d^4,a b^2,b^3,b^2 c,b^2 d,b c d,b d^2,\nonumber \\
c d^2,d^3,a b,a d,b^2,b
   c,b d,c d,d^2,a,b,c,d,1\},
\end{gather}
Note that for $5$-point kinematics, there exists no vector $\omega$ perpendicular to all
external legs. So it is not obvious to find spurious terms directly from the
integrand basis. However, we have the following identities,
\begin{eqnarray}
  \label{dbox5_parity}
  \int \frac{d^4 l_1}{(2\pi)^2}\frac{d^4
    l_2}{(2\pi)^2}\frac{\epsilon(l_1,l_2,p_1,p_2)g(l_2)}{D_1 \ldots D_7} =0,
  \\
\int \frac{d^4 l_1}{(2\pi)^2}\frac{d^4
    l_2}{(2\pi)^2}\frac{\epsilon(l_2,l_1,p_4,p_5)f(l_1)}{D_1 \ldots D_7} =0,
\end{eqnarray}
because of the parity properties for the sub-diagrams. Here $f(l_1)$
and $g(l_2)$ are arbitrary Lorentz-invariant functions of $l_1$ and
$l_2$, respectively. 

There are $6$ branches for cut solutions,
\begin{equation}
  \label{eq:26}
  I=I_1 \cap I_2 \cap I_3 \cap I_4 \cap I_5 \cap I_6.
\end{equation}
Similarly, Define $\omega=d D_1 \wedge \ldots d D_7$. By solving congruence equations, we obtain rank-$7$ forms
$\eta_i$, $i=1, \ldots 6$ such that,
\begin{eqnarray}
  \label{eq:48}
  [\eta_i]_j=\delta_{ij} [\Omega]_j,\quad 1\leq i,j \leq 6.
\end{eqnarray}
We find that each of the first $4$ differential forms $\eta_1, \ldots, \eta_4$
generates $3$ IBPs, while each of the differential forms $\eta_5$ and
$\eta_6$ generate $4$ IBPs. These relations are linearly independent,
so there are $24$ IBPs in total. Furthermore, the identities
(\ref{dbox5_parity}) provides two more independent identities. So we
have $32-26=6$ integrals left,
\begin{gather}
  \label{eq:58}
  I_{\text{turtle}}[1],\quad I_{\text{turtle}}[l_1 \cdot p_4],\quad
  I_{\text{turtle}}[l_1 \cdot p_5],\quad  I_{\text{turtle}}[l_2 \cdot
  p_1],\nonumber \\ \quad  I_{\text{turtle}}[l_2 \cdot p_2],
\quad  I_{\text{turtle}}[(l_1\cdot p_4)( l_2 \cdot p_2)]
\end{gather}

There is a subtlety for the master integrals of turtle diagram. For
the $D$-dimensional cases, there are $3$ master integrals,
$I_{\text{turtle}}[1]$, $I_{\text{turtle}}[l_1 \cdot p_4]$ and
$I_{\text{turtle}}[l_1 \cdot p_5]$. However, for $D=4$, there are only
$2$ master integral $I_{\text{turtle}}[1]$, $I_{\text{turtle}}[l_1
\cdot p_4]$, because of an integrand reduction relation in
$4D$. Since we start with the $4D$ minimal integrand, this additional
relation is already incorporated.  Then using $4$ additional IBPs from
\FIRE{} \cite{FIRE},
\begin{eqnarray}
  \label{eq:59}
  I_{\text{turtle}}[l_2 \cdot
  p_1]&=& I_{\text{turtle}}[l_1 \cdot
  p_5],\nonumber \\
  I_{\text{turtle}}[l_2 \cdot
  p_2]&=& \frac{s_{25}}{s_{14}} I_{\text{turtle}}[l_1 \cdot
  p_4],\nonumber \\
   I_{\text{turtle}}[(l_1\cdot p_4)(l_2 \cdot
  p_2)]&=&\frac{s_{12}s_{45}}{8}I_{\text{turtle}}[1] +\frac{s_{25}}{4}I_{\text{turtle}}[l_1 \cdot
  p_4]-\frac{s_{24}}{4} I[l_1 \cdot p_5], \nonumber \\
  I_{\text{turtle}}[(l_1\cdot p_5)(l_2 \cdot
  p_2)]&=& \frac{s_{15} s_{25}}{4 s_{14}} I_{\text{turtle}}[l_1 \cdot
  p_4] - \frac{s_{25}}{4} I_{\text{turtle}}[l_1 \cdot
  p_5].
\end{eqnarray}
Including these missing IBP relations, we reduce all integrand terms
to the master integrals $I_{\text{turtle}}[1]$, $I_{\text{turtle}}[l_1
\cdot p_4]$. For example,
\begin{gather}
  \label{eq:62}
   I_{\text{turtle}}[l_1\cdot p_5]=-\frac{4 s_{15} \left(s_{12}+s_{15}-s_{34}\right)}{F}  I_{\text{turtle}}[(l_1\cdot p_4)]
  \nonumber \\  -
  \frac{ s_{15} \left(s_{23} s_{34}+\left(s_{15}-s_{34}\right) s_{45}+s_{12}
   \left(s_{15}-s_{23}+2 s_{45}\right)\right) }{F}
I_{\text{turtle}}[1] +\ldots, \\
 I_{\text{turtle}}[(l_1\cdot p_4)(l_2 \cdot p_1)]=-\frac{1}{2F} s_{15} \big(s_{23} s_{34}+(s_{15}-s_{34})
   s_{45}+s_{12} \left(s_{15}-s_{23}+2 s_{45}\right)\big)
   I_{\text{turtle}}[(l_1\cdot p_4)]  \nonumber \\
-\frac{1}{4F} s_{15}
   \left(s_{15}-s_{23}+s_{45}\right)\big(s_{23}
   s_{34}+\left(s_{15}-s_{34}\right) s_{45}+s_{12} \left(s_{15}-s_{23}+2
   s_{45}\right)\big)   I_{\text{turtle}}[1] \nonumber \\+\ldots, \\
 I_{\text{turtle}}[(l_1\cdot p_4)^2(l_1 \cdot p_5)] = -s_{15} \left(s_{12}+s_{15}-s_{34}\right)
   \left(s_{15}-s_{23}+s_{45}\right){}^2 I_{\text{turtle}}[l_1\cdot
   p_4] \nonumber \\
-\frac{1}{4}
s_{15} \left(s_{15}-s_{23}+s_{45}\right){}^2 \big(s_{23}
   s_{34}+\left(s_{15}-s_{34}\right) s_{45}+s_{12} \left(s_{15}-s_{23}+2
   s_{45}\right)\big) I_{\text{turtle}}[1]  +\ldots ,\\
 I_{\text{turtle}}[(l_1\cdot p_4)(l_2 \cdot p_1)(l_2 \cdot p_2)] =0 + 
 \ldots ,
\end{gather}
where the polynomial $F$ is,
\begin{equation}
  \label{eq:82}
 F= 2 \left(2 s_{15}^2+\left(-2 s_{23}-2 s_{34}+s_{45}\right) s_{15}+s_{12}
   \left(s_{15}-s_{23}\right)+s_{34} \left(s_{23}-s_{45}\right)\right)
\end{equation}

The complete result for $4D$ on-shell turtle box IBPs can be downloaded
  at \url{http://www.nbi.dk/~zhang/IBP/dbox5_IBP_result.nb}.

It is interesting to compare our result to the result from GKK
method \cite{GKK_turtle}. GKK method determines that in
$D=4-2\epsilon$ dimension, there are $15$ IBP generating vectors
$v^{(i)}_\text{GKK}$, $i=1,\ldots 15$, 
without doubled propagator. However, in the $4D$ on-shell limit, we explicitly
verified that on each of the $6$ branches, for all $15$ vectors the dual form
$\omega^{(i)}_\text{GKK}$ is proportional to $\Omega$. Hence, in the $4D$
on-shell limit, the $15$ vectors are generated by our six local forms
$\eta_j$, $j=1,\ldots,6$.

\section{Conclusion}
\label{conclusion}
In this paper, we invent a new method to generate integration-by-part
identities from the viewpoint of differential geometry. The
generating vector for IBP identities are reformulated as differential
forms, via Poincar\'{e} dual. Then by techniques of differential
geometry, the geometric meaning of generating
vectors for
IBPs without doubled propagator is clear: {\it they are dual to
the normal direction of the unitarity-cut solution.}

By using the wedge product and congruence
equations over cut branches, suitable differential forms to generate
IBP without doubled propagator are obtained. Our algorithm is
realized by our computational algebraic geometry package, {\sc
  MathematicaM2}.   

We tested our algorithm on several $4D$ two-loop examples. The
algorithm is very efficient in generating the analytic on-shell part of IBP
identities. For example, our program obtains the analytic on-shell
IBPs of 5-point turtle diagram, in about one hour on our laptop. 

Following our discoveries, there are several interesting future directions,
\begin{itemize}
\item The extension of our formalism to
  $D=(4-2\epsilon)$-dimension. Apparently, the differential forms are
  not directly defined in non-integer dimensions. But we expect that
  this difficulty can be circumvented by considering our formalism in
  various integer-valued dimensions, and then combine the results by
  an analytic continuation. In general, the $D$-dimensional unitarity
  cut solution has a simpler structure than its $4D$ counterpart, so we
  expect that the discussion on the local properties of differential
  forms can be simplified in $D$-dimensional cases.
\item The beyond-on-shell part of IBP. For the purpose of finding the
  contour weights in maximal unitarity \cite{Kosower:2011ty}, the algorithm is enough since
  it aims at the on-shell part. It is interesting to see that how to
  go steps further by releasing the cut constraints recursively.
\item Combination of our differential form method with the classic IBP
  generating algorithm like Laporta. Our method focuses on the IBP
  relations without doubled propagator, while other algorithms can
  recover all the IBP relations. Even before applying the
  sophisticated congruence method, it is straightforward to
  calculate the differential form $\Omega = dD_1 \wedge \ldots \wedge dD_k$ analytically, 
and this form itself generate a lot of IBPs without doubled
propagator. We expect that the ingredients of our method can be
  incorporated current IBP generating programs to speed up the computation.
\end{itemize}

\section*{Acknowledgement}
We thank Simon Badger, Emil
J. Bjerrum-Bohr, Spencer Bloch, Simon Caron-Huot, Poul Damgaard, Hjalte
Frellesvig, Rijun Huang, David Kosower, Kasper Larsen and Mads S\o gaard for useful discussion on this
project.  We express special gratitude to Simon Caron-Huot for
his participance in the early stage of this paper and careful reading
of this paper in the draft stage. We also thank David Kosower and
IPhT, Saclay for the hospitality during YZ's visit. YZ is supported by Danish Council for
Independent Research-Natural Science (FNU) grant 11-107241.

\appendix
\section{Review of mathematical notations}

The denominators $D_1, \ldots D_k$ for a Feynman integral, generates an ideal in the polynomial
ring $R=\mathbb C[x_1, \ldots x_{DL}]$,
\begin{equation}
  \label{eq:1}
  I= \langle D_1, \ldots D_k \rangle .
\end{equation}
The {\it cut solution} is the zero locus of all denominators,
\begin{equation}
  \label{eq:2}
  \mathcal S=\mathcal Z(I)=\{(a_1, \ldots a_{DL})\in \mathbb C^{LD}|D_1(a_1, \ldots a_{DL})=\ldots=D_k(a_1, \ldots a_{DL})=0\}.
\end{equation}
In many cases, the cut solution contains several branches, in
mathematical language, the ideal $I$ has a primary decomposition,
\begin{equation}
  \label{primary decomposition}
  I=I_1 \cap \ldots \cap I_n,
\end{equation}
So correspondingly, the cut solution decomposes into several
irreducible branches,
\begin{equation}
  \label{eq:7}
  \mathcal S=\mathcal S_1 \cup \ldots \cup \mathcal S_n,
\end{equation}
where $\mathcal S_j=\mathcal Z(I_j)$.

By Hilbert's Nullstellensatz, if a polynomial $f$ vanishes everywhere on
$\mathcal Z(I)$, then $f\in \sqrt I$. Here $\sqrt I$ is the radical of
$I$,
\begin{equation}
  \label{eq:5}
  \sqrt I=\{f|f^s\in R, s\in \mathbb N\}.
\end{equation}
$\sqrt I$ is also an ideal and $I\subset \sqrt I$. If $I=\sqrt I$, we
call $I$ a radical ideal. 

The integrand $N$ can be reduced by polynomial division towards the
denominators, via {Gr\"obner basis}
\begin{equation}
  \label{eq:81}
  N= \Delta + \sum_i^k f_i D_i,
\end{equation}
where the remainder $\Delta$, is the {\it integrand basis}. We call
monomials in $\Delta$ {\it irreducible numerators}.

\end{document}